\documentclass[sigconf]{acmart}

\usepackage{fancyhdr}
\usepackage{booktabs} 
\usepackage{amssymb}
\usepackage{amsmath,amsthm}
\usepackage{graphicx}
\usepackage[tight,footnotesize]{subfigure}
\usepackage{multirow}
\usepackage{listings}
\usepackage{color}
\usepackage{caption}
\usepackage{algorithm}
\usepackage{algorithmicx}
\usepackage[noend]{algpseudocode}
\usepackage{xspace}
\usepackage{url}
\usepackage{breakurl}

\usepackage{flushend}
\newtheorem{theorem}{Theorem}

\newcommand{\sysname}{D-Spot\xspace}

\newcommand{\para}[1]{\smallskip\noindent\textbf{#1}\xspace}

\newcommand{\subb}[1]{\mathcal{B}}

\newcommand{\user}[1]{u_{#1}}

\newcommand{\subg}{\mathcal{G}}
\newcommand{\subghat}{\hat{\subg}}
\newcommand{\subv}{\mathcal{V}}
\newcommand{\sube}{\mathcal{E}}

\newcommand{\sscoret}{$\mathcal{S}$-score}
\newcommand{\sscore}[2]{\mathcal{S}_{#1,#2}}
\newcommand{\sscoreij}{\mathcal{S}_{i,j}}
\newcommand{\sscorei}{\mathcal{S}_{i}}

\newcommand{\fscore}[1]{\mathcal{F}_{#1}}
\newcommand{\fscoreg}{\mathcal{F}_{\subg}}

\newcommand{\iscorek}[3]{\mathcal{I}_{#1,#2}^{#3}}

\newcommand{\pkx}{p^{k}(a)}

\copyrightyear{2019}
\acmYear{2019} 
\setcopyright{iw3c2w3}
\acmConference[WWW '19]{Proceedings of the 2019 World Wide Web Conference}{May 13--17, 2019}{San Francisco, CA, USA}
\acmBooktitle{Proceedings of the 2019 World Wide Web Conference (WWW '19), May 13--17, 2019, San Francisco, CA, USA}
\acmPrice{}
\acmDOI{10.1145/3308558.3313403}
\acmISBN{978-1-4503-6674-8/19/05}

\usepackage{balance}

\begin{document}
	
\title{No Place to Hide: Catching Fraudulent Entities in Tensors}

\author{Yikun Ban}
\authornote{Yikun did the project during his visit at Tsinghua University.}
\affiliation{%
  \institution{Peking University}
  \streetaddress{P.O. Box 1212}
  \postcode{100084}
}
\email{banyikun@pku.edu.cn}

\author{Xin Liu}
\affiliation{%
	\institution{Tsinghua University}
	\streetaddress{P.O. Box 1212}
	\postcode{100084}
}
\email{liuxin16@mails.tsinghua.edu.cn}


\author{Yitao Duan}
\orcid{1234-5678-9012}
\affiliation{%
	\institution{Fintec.ai}
	\streetaddress{P.O. Box 1212}
	\postcode{100084}
}
\email{duan@fintec.ai}

\author{Xue Liu}
\affiliation{%
	\institution{McGill University}
	\streetaddress{P.O. Box 1212}
	\postcode{}
}
\email{xueliu@cs.mcgill.ca}

\author{Wei Xu}
\affiliation{%
	\institution{Tsinghua University}
	\streetaddress{P.O. Box 1212}
	\postcode{100084}
}
\email{weixu@tsinghua.edu.cn}


 \vspace{-0.5em}

\pagestyle{fancy}

\begin{abstract}
Many approaches focus on detecting dense blocks in the tensor of multimodal data to prevent fraudulent entities (e.g., accounts, links) from retweet boosting, hashtag hijacking, link advertising, etc. However, no existing method is effective to find the dense block if it only possesses high density on a subset of all dimensions in tensors.
In this paper, we novelly identify dense-block detection with dense-subgraph mining, by modeling a tensor into a weighted graph without any density information lost. Based on the weighted graph, which we call information sharing graph (ISG), we propose an algorithm for finding multiple densest subgraphs, D-Spot, that is faster (up to 11x faster than the state-of-the-art algorithm) and can be computed in parallel. In an N-dimensional tensor, the entity group found by the ISG+D-Spot is at least 1/2 of the optimum with respect to density, compared with the 1/N guarantee ensured by competing methods. We use nine datasets to demonstrate that ISG+D-Spot becomes new state-of-the-art dense-block detection method in terms of accuracy specifically for fraud detection.

\end{abstract}


\maketitle

 \vspace{-0.5em}
\section{Introduction}

Fraud represents a serious threat to the integrity of social or review networks such as Twitter and Amazon, with people introducing fraudulent entities (e.g., fake accounts, reviews, etc.) to gain more publicity/profit over a brief period. For example, on a social network or media sharing website, people may wish to enhance their account's popularity by illegally buying more followers~\cite{Shah2014Spotting}; on e-commerce websites, fraudsters may register multiple accounts to benefit from ``new user'' promotions.

Consider the typical log data generated from a social review site (e.g., Amazon), which contains four-dimensional features: users, products, timestamps, rating scores. These data are often formulated as a tensor, in which each dimension denotes a separate feature and an entry (tuple) of the tensor represents a review action. Based on previous studies~\cite{crossspot, MZOOM}, fraudulent entities form dense blocks (sub-tensors) within the main tensor, such as when a mass of fraudulent user accounts create an enormous number of fake reviews for a set of products over a short period.  Dense-block detection has also been applied to network intrusion detection \cite{MultiAspectForensics, MZOOM}, retweet boosting detection~\cite{crossspot}, bot activities detection \cite{MZOOM}, and genetics applications~\cite{Saha2010Dense, MultiAspectForensics}.

 Various dense-block detection methods have been developed. One approach uses tensor decomposition, such as CP decomposition and higher-order singular value decomposition~\cite{MultiAspectForensics}. However, as observed in \cite{DCUBE}, such methods are outperformed by search-based techniques~\cite{crossspot, DCUBE, MZOOM} in terms of accuracy, speed, and flexibility regarding support for different density metrics. Furthermore, \cite{DCUBE, MZOOM} provide an approximation guarantee for finding the densest/optimal block in a tensor.
 
We have examined the limitations of search-based methods for dense-block detection. First, these methods are incapable of detecting \emph{hidden-densest blocks}. We define a hidden-densest block as one that does not have a high-density signal on all dimensions of a tensor, but evidently has a high density on a subset of all dimensions. Moreover, existing methods neglect the data type and distribution of each dimension on the tensor. Assuming that two dense blocks A and B have the same density, however, A is the densest on a subset of critical features, such as IP address and device ID, whereas B is the densest on some trivial features such as age and gender. Can we simply believe that A is as suspicious as B? Unfortunately, the answer when using existing methods is `yes.'

To address these limitations, we propose a dense-block detection framework and focus on entities that form dense blocks on tensors. The proposed framework is designed using a novel approach. Given a tensor, the formation of dense blocks is the result of value sharing (the behavior whereby two or more different entries share a distinct value (entity) in the tensor). Based on this key point, we propose a novel \emph{Information Sharing Graph (ISG)} model, which accurately captures each instance of value sharing. The transformation from dense blocks in a tensor to dense subgraphs in ISG leads us to propose a fast, high-accuracy algorithm, D-Spot, for determining fraudulent entities with a provable guarantee regarding the densities of the detected subgraphs.

In summary, the main contributions of this study are as follows: 

1) \textbf{[Graph Model].} We propose the novel ISG model, which converts every value sharing in a tensor to the representation of weighted edges or nodes (entities). 
Furthermore, our graph model considers diverse data types and their corresponding distributions based on information theory to automatically prioritize multiple features.

2) \textbf{[Algorithm].}  We propose the \emph{D-Spot} algorithm, which is able to find multiple densest subgraphs in one run. And we theoretically prove that the multiple subgraphs found by D-Spot must contain some subgraphs that are at least 1/2 as dense as the optimum. In real-world graphs, D-Spot is up to $11 \times$ faster than the state-of-the-art competing algorithm.

3) \textbf{[Effectiveness].}  In addition to dense blocks, ISG+D-Spot also effectively differentiates hidden-densest blocks from normal ones. 
In experiments using eight public real-world datasets, ISG+D-Spot detected fraudulent entities more accurately than conventional methods.

 \vspace{-0.5em}

\section{Background}

\subsection{Economics of Fraudsters}

As most fraudulent schemes are designed for financial gain, it is essential to understand the economics behind the fraud.  Only when the benefits to a fraudster outweigh their costs will they perform a scam. 
 
To maximize profits, fraudsters have to share/multiplex different resources (e.g., fake accounts, IP addresses, and device IDs) over multiple frauds.  
For example,~\cite{Jiang2014CatchSync} found that many users are associated with a particular group of followers on Twitter;~\cite{thomas2014dialing} identified that many cases of phone number reuse;~\cite{spamscatter} observed that the IP addresses of many spam proxies and scam hosts fall into a few uniform ranges; and~\cite{COPYCATCH} revealed that fake accounts often conduct fraudulent activities over a short time period.

Thus, fraudulent activities often form dense blocks in a tensor (as described below) because of this resource sharing.
 \vspace{-0.5em}

\subsection{Related Work}

\para{Search-based dense-block detection in tensors.} Previous studies~\cite{MultiAspectForensics, crossspot, MZOOM} have shown the benefit of incorporating features such as timestamps and IP addresses, which are often formulated as a multi-dimensional tensor.    
Mining dense blocks with the aim of maximizing a density metric on tensors is a successful approach. CrossSpot~\cite{crossspot} randomly chooses a seed block and then greedily adjusts it in each dimension until the local optimum is attained. This technique usually requires enormous seed blocks and does not provide any approximation guarantee for finding the global optimum. In contrast to adding feature values to seed blocks,  M-Zoom~\cite{MZOOM} removes feature values from the initial tensor one by one using a similar greedy strategy, providing a $1/N$-approximation guarantee for finding the optimum (where $N$ is the number of dimensions in the tensor). M-Biz~\cite{MBIZ} also starts from a seed block and then greedily adds or removes feature values until the block reaches a local optimum.  Unlike M-Zoom, D-Cube~\cite{DCUBE} deletes a set of feature values on each step to reduce the number of iterations, and is implemented in a distributed disk-based manner. D-Cube provides the same approximation guarantee as M-Zoom.

\para{Tensor decomposition methods.} Tensor decomposition \cite{tensorreview} is often applied to detect dense blocks within tensors \cite{MultiAspectForensics}. Scalable algorithms, such as those described in~\cite{Papalexakis2012ParCube, Wang2015Fast, Shin2015Distributed}, have been developed for tensor decomposition. However, as observed in \cite{crossspot,DCUBE}, these methods are limited regarding the detection of dense blocks, and usually detect blocks with significantly lower densities, provide less flexibility with regard to the choice of density metric, and do not provide any approximation guarantee.

\para{Dense-subgraph detection.} A graph can be represented by a two-dimensional tensor, where an edge corresponds to a non-zero entry in the tensor. The mining of dense subgraphs has been extensively studied \cite{GraphSurvey}. Detecting the densest subgraph is often formulated as finding the subgraph with the maximum average degree, and may use exact algorithms~\cite{Goldberg1984Finding,onfinddense} or approximate algorithms~\cite{greedy, onfinddense}. Fraudar~\cite{FRAUDAR} is an extended approximate algorithm that can be applied to fraud detection in social or review graphs. CoreScope~\cite{CoreScope} tends to find dense subgraphs in which all nodes have a degree of at least $k$.   Implicitly, singular value decomposition (SVD) also focuses on dense regions in matrixes. EigenSpoke~\cite{SPOKEN} reads scatter plots of pairs of singular vectors to find patterns and chip communities, \cite{Chen2012Dense} extracts dense subgraphs using a spectral cluster framework, and \cite{Shah2014Spotting,Jiang2016Inferring} use the top eigenvectors from SVD to identify abnormal users.

\para{Other anomaly/fraud detection methods} 
The use of belief propagation \cite{FRAUDEAGLE, NETPROBE} and HITS-like  ideas \cite{COMBATING,Jiang2014CatchSync,UNDERSTAN} is intended to catch rare behavior patterns in graphs. Belief propagation has been used to assign labels to the nodes in a network representation of a Markov random field \cite{FRAUDEAGLE}.  When adequate labeled data are available, classifiers can be constructed based on multi-kernel learning~\cite{Abdulhayoglu2017HinDroid}, support vector machines \cite{Tang2009Machine}, and $k$-nearest neighbor \cite{KNEAR} approaches.

\vspace{-0.5em}

\section{Definitions and Motivation}
In this section, we introduce the notations and definitions used throughout the paper, analyze the limitations of existing approaches, and describe our key motivations.

\vspace{-0.5em}
\begin{table}
\footnotesize
\caption{ ISG+D-Spot vs. existing dense-block detection methods.}
\vspace{-0.7em}
\begin{tabular}{c|ccccc|l}
\toprule
  &\rotatebox{90}{MAF \cite{MultiAspectForensics}}  &\rotatebox{90}{CrossSpot \cite{crossspot}} &\rotatebox{90}{M-zoom \cite{MZOOM}}& \rotatebox{90}{M-biz \cite{MBIZ}} &\rotatebox{90}{D-cube\cite{DCUBE}} & \rotatebox{90}{ISG+D-Spot} \\
\midrule
Applicable to N-dimensional data?&$\surd$&$\surd$&$\surd$&$\surd$&$\surd$&$\surd$\\
Catch densest blocks? & $\surd$ & $\surd$ & $\surd$ & $\surd$ & $\surd$ & $\surd$  \\
Catch hidden-densest blocks? &$\times$ & $\times$ & $\times$ & $\times$ & $\times$ &$\surd$ \\
 \%-Approximation Guarantee? & $\times$& $\times$ &1/N& $\times$ & 1/N & 1/2  \\
\bottomrule
\end{tabular}
\centering
\vspace{-1em}
\end{table}

 \vspace{-0.5em}
\begin{table}
\footnotesize
	\caption{Symbols and Definitions}
	 \vspace{-1em}
	\centering
	\label{tab:freq}
	\begin{tabular}{c|cl}
		\toprule
		Symbol & Interpretation\\
		\midrule
		$N$ & number of dimensions in a tensor\\
		$[N]$ & set \{1,..., N\}\\
		$\mathbf{R}(A_1, ... , A_N, X)$&  relation representing a tensor\\
		$A_n$& $n$-th dimensional values of $\mathbf{R}$\\
		$\mathbf{R}_n$& set of distinct values of $A_n$ of $\mathbf{R}$\\ 
		$t = (a_1, ... , a_N, x)$ & an entry (tuple) of $\mathbf{R}$\\ 
		$ \mathcal{B}(\mathcal{A}_1, ..., \mathcal{A}_N, \mathcal{X})$ & a block in $\mathbf{R}$\\
		$\mathcal{B}_n$& set of distinct values of $\mathcal{A}_n$ of $\mathcal{B}$\\
		$U$ & target dimension in $\mathbf{R}$\\
		$\mathbf{V} = \{u_1, ... , u\}$& set of distinct values of $U$\\  
		$\mathbf{G}=(\mathbf{V}, \mathbf{E})$& Information Sharing Graph \\
		$\sscoreij$ & $\mathcal{S}$-score between $u_i$ and $u_j$\\
		$\sscorei$ & $S$-score of $u_i$\\
		$\subg = (\subv, \sube)$ & subgraph in $\mathbf{G}$\\
 		 $\mathcal{F}$ & density metric \\
		\bottomrule
	\end{tabular}
	 \vspace{-1em}
\end{table}

\subsection{Notation and Formulations}
Table 2 lists the notations used in this paper.
We use $[N] = \{1,..., N\}$ for brevity. 
Let $\mathbf{R}( A_1, ..., A_N, X) =  \{t_0, ... , t_{|X|}\}$ be a relation with $N$-dimensional features, denoted by $\{A_1, ..., A_N\}$, and a dimensional entry identifiers, denoted by $X$.
For each entry (tuple) $t \in \mathbf{R}$, $t = (a_1, ... , a_N, x)$, where  $\forall n \in [N]$, we use $t[A_n]$ to denote the value of $A_n$ in $t$, $t[A_n] = a_n$ and $t[X]$ to denote the identification of $t$, $t[X] = x$, $x \in X$. 
We define the mass of $\mathbf{R}$ as $|\mathbf{R}|$, which is the total number of such entries, $|\mathbf{R}| = |X|$.
For each $n \in  [N]$,  we use $\mathbf{R}_n$ to denote the set of distinct values of $A_n$. 
Thus, $\mathbf{R}$ naturally represents an $N$-dimensional tensor of size $|\mathbf{R}_1| \times ... \times |\mathbf{R}_N|$. 

A block $\mathcal{B}$ in $\mathbf{R}$ is defined as $\mathcal{B}(\mathcal{A}_1, ..., \mathcal{A}_N, \mathcal{X}) = \{t \in \mathbf{R}: t[X] \in \mathcal{X}\}$ and $\mathcal{X} \subseteq X$. Additionally, the mass $|\mathcal{B}|$ is the number of entries of $\mathcal{B}$ and $\mathcal{B}_n$ is the set of distinct values of $\mathcal{A}_n$. Let $\mathcal{B}(a, A_n) =\{t \in \mathbf{R}: t[A_n] = a\}$ represent all entries that take the value $a$ on $A_n$. The mass $|\mathcal{B}(a, A_n)|$ is the number of such entries. A simple example is given as follows.
\smallskip

\noindent \textbf{Example 1} (Amazon review logs). \textit{Assume a relation $\mathbf{R}(\underline{user} , \underline{product},$ $ \underline{timestamp}, X)$, where $\forall t \in \mathbf{R}$,  $t = (a_1, a_2, a_3, x)$ indicates a review action where user $a_1$ reviews product $a_2$ at timestamp $a_3$, and the identification of the action is $x$. Because $a_1$ may review $a_2$ at $a_3$ (we assume that $a_3$ represents a period) multiple times, $X$ helps us distinguish each such action.
The mass of $\mathbf{R}$, denoted by $|\mathbf{R}|$, is the number of all review actions in the dataset. The number of distinct users in $\mathbf{R}$ is $|\mathbf{R}_1|$.
A block $\mathcal{B}(a_1, user)$ is the set of all rating entries operated by user $a_1$, and the number of overall entries of $\mathcal{B}(a_1,user)$ is $|\mathcal{B}(a_1, user)|$}. 
\smallskip

First, we present a density metric that is known to be useful for fraud detection \cite{MZOOM, DCUBE}: \smallskip
 \vspace{-0.5em}
\begin{definition}  (Arithmetic Average Mass $\rho$). 
Given a block \\ $\mathcal{B}(\mathcal{A}_1, ... , \mathcal{A}_N, \mathcal{X})$, the arithmetic average mass of $\mathcal{B}$ on dimensions $\mathcal{N}$ is 
\begin{displaymath}
\rho(\mathcal{B}, \mathcal{N}) = \frac{|\mathcal{B}|}{ \frac{1}{|\mathcal{N}|} \sum_{n \in \mathcal{N}} |\mathcal{B}_n|},
\end{displaymath}
where $\mathcal{N}$ is a subset of $[N]$ and obviously $\rho \in [1.0, +\infty)$.
\end{definition}
 \vspace{-0.5em}

\emph{ If block $\mathcal{B}$ is dense in $\mathbf{R}$, then $\rho(\mathcal{B}, [N])> 1.0$.}

Other density metrics listed in \cite{DCUBE} are also effective for fraud detection. It is broadly true that all density measures are functions of the cardinalities of the dimensions and masses of $\mathcal{B}$ and $\mathbf{R}$. In $\mathbf{R}$, previous studies~\cite{crossspot,DCUBE,MBIZ,MZOOM} have focused on detecting the top-$k$ densest blocks in terms of a density metric. In the remainder of this paper, we use the density metric $\rho$ to illustrate our key points.


\subsection{Shortcomings of Existing Approaches and Motivation}

In practice, the blocks formed by fraudulent entities in $\mathbf{R}$ may be described by \emph{hidden-densest blocks}. To illustrate hidden-densest blocks, we present the following definitions and examples.\smallskip
 \vspace{-0.5em}
\begin{definition}  In  $\mathbf{R}( A_1, ..., A_N,$ $ X)$, we say that $\mathcal{B}(\mathcal{A}_1, ... ,$ $ \mathcal{A}_N, \mathcal{X})$ is the densest on a dimension $A_n$ if $\rho(\mathcal{B}, \{n\})$ is the maximal value of all possible $\rho(\hat{\mathcal{B}}, \{n\})$, where $\hat{\mathcal{B}}$ is any possible block in $\mathbf{R}$.
\end{definition}
 \vspace{-0.5em}

 \vspace{-0.5em}
\begin{definition} (Hidden-Densest Block). In  $\mathbf{R}( A_1, ..., A_N,$ $ X)$,  $\mathcal{B}(\mathcal{A}_1, ... , $ $\mathcal{A}_N, \mathcal{X})$ is the hidden-densest block if $\mathcal{B}$ is the densest on a small subset of $\{ A_1, ..., A_N \}$.
\end{definition}
 \vspace{-0.5em}

\noindent \textbf{Example 2} (Registration logs). \textit{In a registration dataset with 19 features, fake accounts only exhibit conspicuous resource sharing with respect to the IP address feature.}
\smallskip

\noindent \textbf{Example 3} (TCP dumps). \textit{The DARPA dataset~\cite{AIRFORCE} has 43 features, but the block formed by malicious connections is only the densest on two features.} \smallskip

Thus, catching hidden-densest blocks has a significant utility in the real world. Unfortunately, the problem is intractable using existing approaches \cite{MZOOM,DCUBE,MBIZ,crossspot}. 

First, assuming that the  hidden-densest block  $\mathcal{B}(\mathcal{A}_1, ..., \mathcal{A}_N, \mathcal{X}) $ is only the densest on dimension $A_N$, we have that
\begin{displaymath}
 \rho(\mathcal{B}, [N-1]) = \frac{|\mathcal{B}|}{ \frac{1}{N-1} \sum_{n \in [N-1]} |\mathcal{B}_n|} \approx \rho(\mathcal{B}, [N]) = \frac{|\mathcal{B}|}{ \frac{1}{N} \sum_{n \in [N]} |\mathcal{B}_n|}
\end{displaymath}
when $N$ is sufficiently large. Assuming $\rho(\mathcal{B}, [N-1])$ is very low, then
the methods in \cite{crossspot, MZOOM, DCUBE, MBIZ}, which try to find the block $\mathcal{B}$ that maximizes $\rho(\mathcal{B}, [N])$,  have a limited ability to detect the hidden-densest block.

Second, consider a block $\mathcal{B}$ formed by fraudulent entities, in which $\mathcal{B}$ is only the densest on $\{A_2, A_3, A_5\}$, and thus $\rho(\mathcal{B}, \{2,3,5\})$ is maximal. However,  the techniques of \cite{crossspot, MZOOM,MBIZ,DCUBE} cannot find $\{A_2, A_3, A_5 \}$ when the feature combinations are exploded.

Furthermore, in  $\mathbf{R}( A_1, ..., A_N,$ $ X)$, consider two blocks $\mathcal{B}_1$ and $\mathcal{B}_2$, where $\mathcal{B}_1$  is the densest on $A_i$, $\mathcal{B}_2$ is the densest on $A_j$, and $\rho(\mathcal{B}_1, [N]) = \rho(\mathcal{B}_2, [N])$. Does this indicate that $\mathcal{B}_1$ and $\mathcal{B}_2$ are equally suspicious? No, absolutely not, because $A_i$ could be the IP address feature and $A_j$ could be a trivial feature such as the user's age, location, or gender.    

\para{[Value Sharing]}. Based on the considerations above, we design our approach from a different angle.  The key reason behind the formation of dense blocks is value sharing.  Given $t_1 \in \mathbf{R}$, a dimension $A_n$, and $t_1[A_n] = a$, we can identify value sharing when $\exists t_2 \in \mathbf{R}$, $t_2 \neq t_1$, and $t_2[A_n] = a$.     
\smallskip

Obviously, if a block $\mathcal{B}$ is dense, $\rho(\mathcal{B}, [N]) > 1.0$, then value sharing must be occurring, i.e., \textbf{value sharing results in dense blocks}. 
 
Therefore, detecting dense blocks is equivalent to catching value sharing signals. We propose ISG based on information theory and design the D-Spot algorithm to leverage graph features, allowing us to catch fraudulent entities within dense blocks and overcome the limitations mentioned above.

\begin{figure}
\centering
 \includegraphics[width=1.0\columnwidth]{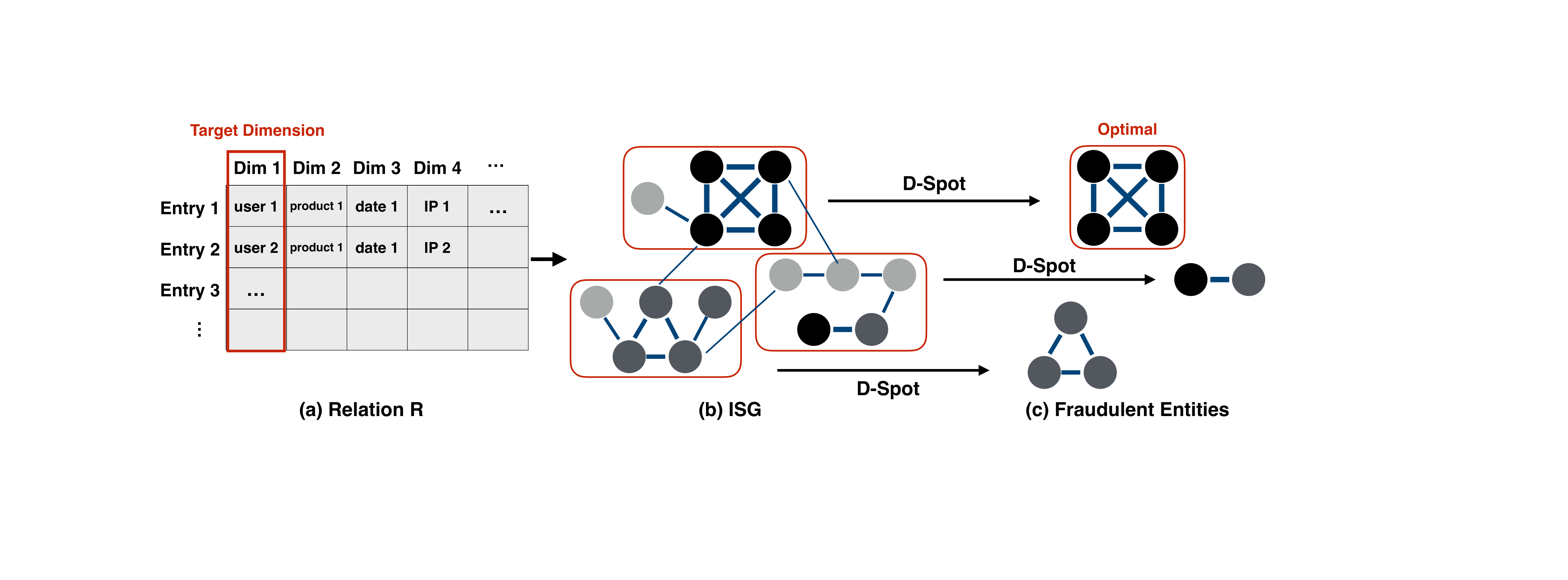}
  \vspace{-0.2cm}
 \caption{Workflow of ISG+D-Spot.} 
 \label{fig:workflow}
 \vspace{-0.4cm}
\end{figure}

\vspace{-0.5em}
\section{ISG Building} \label{sec:method}
In this section, we present the Information Sharing Graph (ISG), which is constructed on the relation $\mathbf{R}$. 

\subsection{Problem Formulation}

Catching fraudulent entities is equivalent to detecting a subset of distinct values in a certain dimension. Let $U$ denote the target dimension in which a subset of distinct values form the fraudulent entities we wish to detect. 
In $\mathbf{R} = ( A_1, ..., A_N, X)$, \textbf{we choose a dimension and set it as $U$, and denote the remaining $(N-1)$ dimensions as $K$ dimensions}, \underline{$k \in [K]$}, for brevity.
We build the ISG of $U$, i.e., the weighted-undirected graph $\mathbf{G}=(\mathbf{V}, \mathbf{E})$, in which $\mathbf{V} = \{u_1, ... , u_n \}$ is the set of distinct values of $U$.

In Example 1, $\mathbf{R}(\underline{user} , \underline{product},$ $ \underline{timestamp}, X)$, we set $U = \underline{user}$ if we wish to detect fraudulent user accounts. In Example 2, we set $U = \underline{account}$ if we would like to identify fake accounts.
In Example 3, we set $U = \underline{connection}$ to catch malicious connections.

To specifically describe the process of value sharing, we present the two following definitions:
 \vspace{-0.5em}
\begin{definition} (Pairwise Value Sharing). Given $u_i, u_j \in \mathbf{V}$ and $ a \in A_k$, we say that $u_i$ and $u_j$ share value $a$ on $A_k$ if $\exists t_1, t_2 \in \mathbf{R}$ such that $t_1[U] = u_i, t_2[U] = u_j$ and $t_1[A_k] = t_2[A_k] = a$.
\end{definition}
 \vspace{-0.5em}

Pairwise value sharing occurs when a distinct value is shared by multiple individual entities. Given a value sharing process in which $a$ is shared by $\mathcal{V} \subset \mathbf{V}$, we denote this as $\frac{|\mathcal{V}|(|\mathcal{V}|-1)}{2}$ pairwise value sharing.

\begin{definition} (Self-Value Sharing). Given $t_1 \in \mathbf{R}$, where $t_1[U] = u_i$, $u_i \in \mathbf{V}$,  and $ t_1[A_k] = a$, we say that $u_i$ shares value $a$ on $A_k$ if $\exists t_2 \in \mathbf{R}$ and $t_2 \neq t_1$ such that $t_2[U] = u_i$ and $ t_2[A_k] = a$.
\end{definition}

Another type of  value sharing occurs when the distinct value $a$ is shared $n$ times by an entity $u_i$, which can be represented by $n$ instances of self-value sharing.

In ISG $\mathbf{G = (V, E)}$, for some edge $(u_i, u_j) \in \mathbf{E}$,  $\sscoreij$ represents the information between $u_i$ and $u_j$ derived from the other $K$ dimensions, and for some node $u_i \in \mathbf{V}$, $\sscorei$ denotes the information of $u_i$ calculated from the other $K$ dimensions. From the definitions and notations defined in the previous section, Problem 1 gives a formal definition of how to build the ISG of a tensor.
\smallskip

\noindent \textbf{Problem 1} (Building a pairwise information graph). 
\textbf{(1) Input:} \textit{a relation $\mathbf{R}$, the target dimension $U$,} \textbf{(2) Output:} \textit{the information sharing graph $\mathbf{G}=(\mathbf{V}, \mathbf{E})$.}

\subsection{Building an ISG} \label{sec:computing}

Given a dimension $A_k$, the target dimension $U$, any $u_i \in \mathbf{V}$, and an entry $t_1 \in \mathbf{R}$ for which $t_1[U] = u_i$, then for each $a \in A_k$, we assume that the probability of $t_1[A_k] = a $ is $\pkx$.

\para{Edge Construction.} \
Based on information theory \cite{INFORMATION}, the self-information of the event that $u_i$ and $u_j$ share $a$ is:
 \vspace{-0.5em}
\begin{equation}
\iscorek{i}{j}{k}(a) = \log (\frac{1}{\pkx})^2. \label{eq:infox}
\end{equation}
To compute the pairwise value sharing between $u_i$ and $u_j$ across all $K$ dimensions, we propose the metric \sscoret \ as the edge weight of ISG:
 \vspace{-0.5em}
\begin{equation}\label{eq:sscore}
\sscore{i}{j}=\sum_{k=1}^{K}\sum_{a  \in H_k(u_i, u_j)} \iscorek{i}{j}{k}(a),
\end{equation}
where $H_k(u_i, u_j)$ is the set of all values shared by $u_i$ and $u_j$ on $A_k$. Note that $\sscore{i}{j} = 0.0$ if $\bigcup_{k=1}^{K} H_k(u_i, u_j) = \emptyset $.

Intuitively, if $u_i$ and $u_j$ \emph{do not} have any shared values, which is to be expected in  normal circumstances, we have zero information.  Otherwise, we obtain some information. Thus, the higher the value of $\sscoreij$ is, the more similar $u_i$ is to $u_j$.  In practice, the \sscoret \ has a large variance. For example, fraud user pairs sharing  an IP subnet and device ID will have a high \sscoret, whereas normal users are unlikely to share these values with anyone, and will thus have an \sscoret \ close to zero. Additionally, the information we obtain for $u_i$ and $u_j$ sharing the value $a$ is related to the overall probability of that value.  For example, it would be much less surprising if they both follow Donald Trump on Twitter than if they both follow a relatively unknown user.  

\para{Node Setting.} For a node $u_i \in \mathbf{V}$,  let $\mathcal{B}(a, A_k, u_i, U)$ be the set $\{t\in \mathbf{R}: (t[A_k] = a) \wedge  (t[U] = u_i)\}$. When $|{\mathcal{B}(a, A_k, u_i, U)}| \geq 2$.  The information of forming $\mathcal{B}(a, A_k, u_i, U)$  is:
 \vspace{-0.5em}
\begin{equation}
\mathcal{I}_i^k(a)  = \log (\frac{1}{\pkx})^{|{\mathcal{B}(a, A_k, u_i, U)}|}. 
\end{equation}
 \vspace{-0.5em}

We now define $\sscorei$ to compute the self-value sharing for $u_i$ across all $K$ dimensions:
 \vspace{-0.5em}
\begin{equation}
\sscorei =\sum_{k=1}^{K}\sum_{\mathcal{B}  \in H_k(u_i)} \mathcal{I}_i^k(a),
\end{equation}
where $H_k(u_i)$ is the set $\{\mathcal{B}(a, A_k, u_i, U), \forall a \in \mathbf{R}_k\}$ and $\mathbf{R}_k$ is the set of distinct values of $A_k$. Note that $\sscorei = 0.0$ if $\bigcup_{k=1}^{K} H_k(u_i) = \emptyset$. 

 In effect, self-value sharing occurs in certain fraud cases. For instance, a fraudulent user may create several fake reviews for a product/restaurant on Amazon/Yelp~\cite{YELP} over a few days. In terms of network attacks~\cite{DARPA}, a malicious TCP connection tends to attack a server multiple times.

\para{Determining [ $\pkx$ ].} \
We can extend the \sscoret \ to accommodate different data types and distributions.

It is difficult to determine $\pkx$, as we do not always know the distribution of $A_k$.  In this case, for dimensions that are attribute features, we assume a uniform distribution and simply set 
\begin{equation}
\pkx = 1 / |\mathbf{R}_k|.
\end{equation}
This approximation works well for many fraud-related properties such as IP subnets and device IDs, which usually follow a Poisson distribution~\cite{crossspot}. 

However, the uniform assumption works poorly for low-entropy distributions, such as the long-tail distribution, which is common in dimensions such as items purchased or users followed.  Low entropy implies that many users behave similarly anyway, independent of frauds.  Intuitively for such distributions, there is no surprise in following a celebrity (head of the distribution), but considerable information if they both follow someone at the tail. 
For example, 20\% of users correspond to more than 80\% of the ``follows'' in online social networks. The dense subgraphs between celebrities and their fans are very unlikely to be fraudulent. 
If feature $A_k$ has a long-tail distribution, its entropy is very low. For example, the entropy of the uniform distribution over 50 values is 3.91, but the entropy of a long-tail distribution with 90\% of probabilities centered around one value is only 0.71.
Therefore, we set $\pkx$ based on the empirical distribution as
\begin{equation}
\pkx =  |{\mathcal{B}(a, A_k)}| / |\mathbf{R}|,
\end{equation}
when the values in $A_k$ have low entropy. We also provide an interface so that users can define their own $\pkx$ function.

\para{Optimization of ISG Construction.} 
 In theory, a graph with $|\mathbf{V}|$ nodes has $O(|\mathbf{V}|^2)$ edges. Naively, therefore, it takes $O(|\mathbf{V}|^2)$ time for graph initialization and traversal. 
 
To reduce the complexity of building the ISG, we use the \emph{key-value} approach. The \emph{key} corresponds to a value $a$ on $A_k$ and the \emph{value} represents the block $\mathcal{B}(a, A_k)$. Let $\mathcal{V} \subseteq \mathbf{V}$  denote the entities that occur in $\mathcal{B}(a, A_k)$.
 As each pair $(u_i, u_j) \in \mathcal{V}$ shares $a$, we increase the value of $\sscoreij$ by $\iscorek{i}{j}{k}(a)$. 
Additionally, for each $u_i \in \mathcal{V}$, there exists some $\mathcal{B}(a, A_k, u_i, U) \subseteq \mathcal{B}(a, A_k)$. Thus, we increase the value of $\sscorei$ by $\mathcal{I}_i^k(a)$ if $|{\mathcal{B}(a, A_k, u_i, U)}| \geq 2$.

To build the ISG, we compute all \emph{key-value} pairs across $K$ dimensions by traversing $\mathbf{R}$ in parallel. Thus, it takes $O(K|\mathbf{R}|+|\mathbf{E}|)$ time to build the graph $\mathbf{G = (V, E)}$. Note that we only retain positive $\sscoreij$ and $\sscorei$. In practice, $\mathbf{G}$ is usually sparse, which is discussed in Section \ref{sec:com}.


 \vspace{-1em}
\subsection{Key Observations on ISG }\label{sec:key_obs}

Given a relation $\mathbf{R} = (A_1, ... , A_N, X)$ in which we set $U = A_N$, we construct the ISG of $U$, $\mathbf{G} = (\mathbf{V}, \mathbf{E})$. Assuming there is a \textbf{fraudulent block} $\mathcal{B}$ in $\mathbf{R}$, $\mathcal{B} = (\mathcal{A}_1, ... , \mathcal{A}_N, \mathcal{X})$ is transformed into a subgraph $\mathcal{G} = (\mathcal{V}, \mathcal{E})$ in $\mathbf{G}$, where $\mathcal{V}$ is the set of distinct values of $\mathcal{A}_N$ and an edge $\sscoreij \in \mathcal{E}$ denotes the information between $u_i$ and $u_j$ calculated from the other $K$ dimensions.  Then, $\mathcal{V}$ is the fraud group comprised of fraudulent entities that we wish to detect.

We summarize three critical observations of $\mathcal{G}$ that directly lead to the algorithms presented in Section \ref{sec:bads}.
Given $\mathcal{G} = (\mathcal{V}, \mathcal{E})$, we define the \textbf{edge density} of $\mathcal{G}$ as 
 \vspace{-0.5em}
\begin{displaymath}
\rho_{edge}(\mathcal{G}) = \frac{|\mathcal{E}|}{|\mathcal{V}|(|\mathcal{V}| - 1)}
\end{displaymath}  
 \vspace{-0.5em}

\para{1) The value of $\sscoreij$ or $\sscorei $ is unusually high.} \
Value sharing may happen frequently, but sharing across certain features, even certain values, is more suspicious than others.  Intuitively, it might be suspicious if two users share an IP address or follow the same random ``nobody'' on Twitter.  However, it is not so suspicious if they have a common gender, city, or follow the same celebrity.  In other words, certain value sharing is likely to be fraudulent because the probability of sharing across a particular dimension, or at a certain value, is quite low. Thus, the information value is high, which is accurately captured by $\sscoreij$ and $\sscorei$.

\para{2) $|\mathcal{V}|$ is usually large}. \
Fraudsters perform the same actions many times to achieve economies of scale.  Thus, we expect to find multiple pairwise complicities among fraudulent accounts.  A number of studies have found that large cluster sizes are a crucial indicator of fraud~\cite{COPYCATCH,SYNCHROTRAP}. Intuitively, while it is natural for a few family members to share an IP address, it is highly suspicious when dozens of users share one.

\para{3) The closer $\rho_{edge}(\mathcal{G})$ is to 1.0, the more suspicious $\mathcal{G}$ is.} \
Fraudsters usually operate a number of accounts for the same job, and thus it is likely that users manipulated by the same fraudster will share the same set of values. Thus, the $\mathcal{G}$ formed by the fraud group will be well-connected.


\para{Appearance of legitimate entities on ISG}. \ In $\mathbf{G} = (\mathbf{V}, \mathbf{E})$, given some $u_i$ that we assume to be \textbf{legitimate}, let $h(u_i)$ denote the set of its neighbor nodes. We have two findings.  
(1) For $u_i$, $\sscorei  + \sum_{u_j \in \mathbf{V}} \sscoreij \rightarrow 0 $ because $u_i$ is unlikely to share values with others. Even if exists,  the shared values should have a high probability (see observation 1) and therefore small $\sscoreij$. 
(2) The subgraph $\mathcal{G}$ induced by $h(u_i)$ is typically not well-connected, as resource sharing is uncommon in the real world. If $\mathcal{G}$ is well-connected, $|h(u_i)|$ is quite small compared with the fraud group size (see observation 2).
\smallskip

In summary, the techniques described in \cite{MultiAspectForensics,crossspot,MZOOM,DCUBE,MBIZ} work directly on the tensor, indicating that they consider value sharing on each dimension, and even certain values, as equivalent. In contrast, ISG assigns each instance of value sharing a theoretical weight based on the edges and nodes of the ISG, which is more effective for identifying the (hidden-) densest blocks (comparison in Sec.\ref{sec:eff}).

%
%

\vspace{-0.5em}
\section{Spotting Fraud}

Based on the observations in Section \ref{sec:key_obs}, we now describe our method for finding objective subgraphs in $\mathbf{G}$.
This section is divided into two parts: first, we define a density metric $\fscoreg$, and then we illustrate the proposed D-Spot algorithm. 

\subsection{Density Metric and Problem Definition}
  
To find the objective $\mathcal{G = (V, E)}$, we define a density metric $\fscoreg$ as \cite{greedy, FRAUDAR}:
\vspace{-0.5em}
\begin{equation}\label{eq:fraudscore}
\fscoreg = \frac{\sum_{(\user{i},\user{j}) \in \sube} \sscoreij + \sum_{u_i \in \mathcal{V} } \sscorei }{|\subv|}.
\end{equation}

\noindent The form of $\fscoreg$ satisfies the three key observations of  $\mathcal{G}$ in Section \ref{sec:key_obs}.
\begin{enumerate}
\item Keeping $|V|$ fixed, we have that $\sum_{(u_i, u_j) \in \mathcal{E}}\sscoreij + \sum_{ u_i \in \mathcal{V} } \sscorei \uparrow \Rightarrow \fscoreg \uparrow$.

\item Keeping $\sscoreij$, $\sscorei$, and $\rho_{edge}(\mathcal{G})$ fixed, we have that $|V| \uparrow \Rightarrow \fscoreg \uparrow$.

\item Keeping  $\sscoreij$, $\sscorei$, and $|\mathcal{V}|$ fixed, we have that $\rho_{edge}(\mathcal{G}) \uparrow \Rightarrow \fscoreg \uparrow$.
\end{enumerate}


Thus, our subgraph-detection problem can be defined as follows:

\noindent \textbf{Problem 2} (Detecting dense subgraphs).\ 
\textbf{(1) Input:} \textit{ the information sharing graph $\mathbf{G= (V, E)}$.} \textbf{(2) Find:} \textit{multiple subgraphs of $\mathbf{G}$ that maximize $\mathcal{F}$.}

 \vspace{-0.5em}
\subsection{D-Spot (Algorithm 1-3)} \label{sec:bads}

In real-world datasets, there are usually numerous fraud groups forming multiple dense subgraphs.  Based on the considerations described above, we propose \sysname (Algorithms 1--3). Compared with other well-known algorithms for finding the densest subgraph~\cite{greedy, FRAUDAR}, \sysname has two differences: 
\begin{enumerate}
\item  \sysname can detect multiple densest subgraphs simultaneously. \sysname first partitions the graph, and then detects a single densest subgraph in each partition. 

\item \sysname is faster. 
First, instead of removing nodes one by one,  \sysname removes a set of nodes at once, reducing the number of iterations.
Second, it detects the single densest subgraph in partition $\mathcal{G = (V, E)}$, rather than in graph $\mathbf{G = (V,E)}$, where $|\mathcal{V}| << |\mathbf{V}|$ and $|\mathcal{E}| << |\mathbf{E}|$. 

\end{enumerate}

D-Spot consists of two main steps:
(1) Given $\mathbf{G}$, divide $\mathbf{G}$ into multiple partitions (Algorithm 1); 
(2) In each $\mathcal{G}$, find a single dense subgraph (Algorithms 2 and 3).

\para{Algorithm 1: graph partitioning.} \ 
Let $\subghat$ be a dense subgraph formed by a fraud group that we are about to detect.
In $\mathbf{G}$, there are usually multiple $\subghat$s, where each $\subghat$ should be independent or connected with others by small values of $\sscoreij$. Thus we let $\mathcal{G}$s be the connected components of $\mathbf{G}$ as partitions (line 6).
For each $\mathcal{G} \in \mathcal{G}s$, we run Algorithms 2 and 3 (lines 7--9) to find $\subghat$. Finally, Algorithm 1 returns multiple dense subgraphs $\subghat$s (line 10). Note that there is a guarantee that $\subghat$s must contain the $\subghat$ that is at least $1/2$ of the optimum of $\mathbf{G}$ in terms of $\mathcal{F}$ (proof in Sec.\ref{sec:acc}).

\para{Information pruning (recommended).} As mentioned before, Fraud entities usually have surprising similarities that are quantified by $\sscoreij$. We want to delete edges with regular weights and thus we provide a threshold for removing edges:
\vspace{-0.5em}
\begin{equation}\label{eq:threshold}
\theta = \frac{\sum_{(u_i, u_j) \in \mathbf{E}} \sscoreij }{|\mathbf{V}|(|\mathbf{V}| -1 )}  
 \end{equation}
It is easy to see that $\theta$ (conservative) is the average information of all possible pairs ($u_i, u_j$). Thus, we iterate through all edges in $\mathbf{G}$, and remove those for which $\sscoreij < \theta$ (lines 3--5). In all experiments of this paper, we used $\theta$ and get the expected conclusion that the performance of D-Spot using $\theta$ is hardly different from no pruning but using $\theta$ is able to significantly decrease the running cost of D-Spot.

%
%

\begin{algorithm}[h]
\caption{ \textit{find multiple dense subgraphs} in $\mathbf{G}$} 
\begin{algorithmic}[1]
\Require  
$\mathbf{G} =(\mathbf{V}, \mathbf{E})$, $\theta$ (Eq.\ref{eq:threshold} ), $w()$ (Eq. \ref{eq:w})
\Ensure $\subghat$s
\State $\subghat$s $\leftarrow \emptyset$ 
\If{needed} \For{ each $\sscoreij \in \mathbf{E}$}
\If{ $\sscoreij < \theta$  }

\State remove $\sscoreij$
\EndIf   
\EndFor \EndIf
\State $\mathcal{G}$s $\leftarrow$ connected components of $\mathbf{G}$
 \For{ each $\mathcal{G}$ $\in$ $\mathcal{G}$s}

\State $\subghat \leftarrow$ \textit{find a dense subgraph} ($\mathcal{G}$, $w()$)

\State $\subghat$s $\leftarrow $ $\subghat$s $\cup  \{ \subghat\}$
\EndFor    
\State \Return $\subghat$s 
\end{algorithmic}  
\label{alg:greedy}
\end{algorithm}

\begin{algorithm}[h]
\caption{\textit{find a dense subgraph}} 
\begin{algorithmic}[1]
\Require  $\mathcal{G = (V, E)}$, $w()$(Eq. \ref{eq:w})
\Ensure $\subghat$
\State  $\mathcal{V}_c \leftarrow copy(\mathcal{V})$
\State $\mathcal{S}_{sum} \leftarrow \sum_{(u_i, u_j) \in \mathcal{E}} \sscoreij + \sum_{u_i \in \mathcal{V}}\sscorei$
\State $ \forall u \in \mathcal{V}, Dict1[u] \leftarrow 0, Dict2[u] = w(u, \mathcal{G}) $
\State $ index \leftarrow 0 $, \ $\mathcal{F}^{max} \leftarrow \frac{\mathcal{S}_{sum}}{|\mathcal{V}_c|} $, \ $top \leftarrow 0 $
\While{$\mathcal{V}_c \neq \emptyset$}
\State $R \leftarrow \{ u \in \mathcal{V}_c: Dict2[u] \leq \frac{2 \sum_{(u_i, u_j) \in \mathcal{E}}\sscoreij + \sum_{u_i \in \mathcal{V}_c} \sscorei }{|\mathcal{V}_c|} \} $ (Eq. \ref{eq:avew})
\State sort $R$ in increasing order of $Dict2[u]$
\For{each $u \in R$}
\State $\mathcal{V}_c \leftarrow \mathcal{V}_c - u$, \ $\mathcal{S}_{sum} \leftarrow \mathcal{S}_{sum} - Dict2[u]$
\State $index \leftarrow index + 1 $, $Dict1[u] \leftarrow index$
\State $\mathcal{F} = \frac{\mathcal{S}_{sum}}{|\mathcal{V}_c|}$
\If{$\mathcal{F}>\mathcal{F}^{max}$ }
\State $\mathcal{F}^{max} \leftarrow \mathcal{F}$, \ $top \leftarrow index$
\EndIf
\State $Dict2 \leftarrow  $ \textit{update edges} $(u, \mathcal{V}_c, Dict2, \mathcal{G})$
\EndFor
\EndWhile
\State $\hat{R} \leftarrow \{u \in \mathcal{V}: Dict1[u] > top \}$

\State \Return $\subghat$ (the subgraph induced by $\hat{R}$) 
\end{algorithmic}  
\label{alg:greedy}
\end{algorithm}

\begin{algorithm}[h]
\caption{ \textit{update edges}} 
\begin{algorithmic}[1]
\Require  
$u_i$, $\mathcal{V}_c$, $Dict2$, $\mathcal{G = (V, E)}$
\Ensure $Dict2$
\For{each $u_j \in \mathcal{V}_c$}
\If{$(u_i, u_j) \in \mathcal{E}$}
\State $Dict2[u_j] \leftarrow Dict2[u_j] - \sscoreij$
\State remove $(u_i, u_j)$ from $\mathcal{E}$
\EndIf
\EndFor
\State \Return $Dict2$
\end{algorithmic}  
\end{algorithm}

\para{Algorithms 2 and 3: find a dense subgraph.}  \label{sec:imp}
Initially, let $\mathcal{V}_c$ be a copy of $  \mathcal{V}$. In each iteration (lines 5--14), we delete \textit{a set of nodes} ($R$, line 6) from $\mathcal{V}_c$ until $\mathcal{V}_c$ is empty. Of all the $\mathcal{V}_c$ constructed during the execution of the algorithm, that maximizing $\mathcal{F}$ ($\hat{R}$, line 15) is returned as the output of the algorithm.

Given a subgraph $\mathcal{G = (V, E)}$, for some $u_i \in \mathcal{V}$, we define $w(u_i, \mathcal{G})$ as 
\begin{equation} \label{eq:w}
w(u_i, \mathcal{G}) = \sum_{(u_j \in \mathcal{V}) \wedge ((u_i, u_j) \in \mathcal{E}) } \sscoreij + \sscorei .
\end{equation}

Lines 1--4 initialize the parameters used in the algorithm. $Dict2$ records the $w$ value of each node.
$Dict1$ records the order in which the nodes are deleted (line 10), which allows us to determine the value of $\hat{R}$ that maximizes $\mathcal{F}$.
Line 6 determines which $R$ are deleted in each iteration.
$R$ is confirmed by $\{u \in \mathcal{V}: w(u, \mathcal{G}) \leq \overline{w}\}$ (line 6), where the average $\overline{w}$ is given by:
\begin{equation} \label{eq:avew}
\begin{split}
\overline{w} &= \frac{\sum_{u \in \mathcal{V}} w(u, \mathcal{G})}{|\mathcal{V}|} \\
& = \frac{2 \sum_{(u_i, u_j) \in \mathcal{E}}\sscoreij + \sum_{u_i \in \mathcal{V}} \sscorei }{|\mathcal{V}|} \leq  2 \mathcal{F}_\mathcal{G}, 
\end{split}
\end{equation}   
because each edge $\sscoreij$ is counted twice in $\sum_{u \in \mathcal{V}} w(u, \mathcal{G})$. In lines 7--14, the nodes in $R$ are removed from $\mathcal{V}_c$ in each iteration (In contrast, \cite{FRAUDAR} recomputes all nodes and finds those with the minimal $w$ after deleting a node). As removing a subset of $R$ may result in a higher value of $\mathcal{F}$, D-Spot records each change of $\mathcal{F}$, as if the nodes were removed one by one (lines 8--14).   
Algorithm 3 describes how the edges are updated after a node is removed, requiring a total of $|\mathcal{E}|$ updates. 
Finally, Algorithm 2 returns the subgraph $\subghat$ induced by $\hat{R}$, the set of nodes achieving $\mathcal{F}^{max}$, according to $top$ and $Dict1$ (line 15).

\para{Summary.}  As $R$ contains at least one node, the worst-case time complexity of Algorithm 2 is $O( |\mathcal{V}|^2 + |\mathcal{E}|)$. In practice, the worst case is too pessimistic. In line 6, $R$ usually contains plenty of nodes, significantly reducing the number of scans of $\mathcal{V}_c$ (see Section \ref{sec:speed}).

\vspace{-0.5em}
\section{Analysis}


\subsection{Complexity.} \label{sec:com}
In the graph initialization stage, it takes $O(K|\mathbf{R}|+\mathbf{|E|})$ time to build $\mathbf{G}$ based on the optimization in Section \ref{sec:computing}. In D-Spot,  the cost of partitioning $\mathbf{G}$ is $O(|\mathbf{E}|)$, and detecting a dense block in a partition $\mathcal{G}$ requires $O(|\mathcal{E}| + |\mathcal{V}|^2)$ operations, where $|\sube| << |\mathbf{E}|$, $|\subv| << |\mathbf{V}|$.
Thus, the complexity of  ISG+D-Spot is linear with respect to $|\mathbf{E}|$.

In the worst case, admittedly, $|\mathbf{E}| = |\mathbf{V}|^2$ when there is some dimension $A_k$ in which $|\mathbf{R}_k| = 1$. However, that is too pessimistic. In the target fraud attacks, fraud groups typically exhibit strong value sharing while legit entities should not. Hence, we expect $\mathbf{G}$ to be sparse because the $u_i$ only have positive edges with a small subset of $\mathbf{V}$. 
We constructed a version of $\mathbf{G}$ using several real-world datasets (see Fig. \ref{fig:time}), and the edge densities were all less than 0.06. 


\vspace{-0.5em}
\subsection{ Effectiveness of ISG+D-Spot } \label{sec:eff}

\begin{theorem} (Spotting the Hidden-Densest Block). \label{theo1}
Given a dense block $\mathcal{B}(\mathcal{A}_1, ... ,$ $ \mathcal{A}_N, \mathcal{X})$  in which the target dimension $U = A_N$ and $\mathcal{V}$ denotes the set of distinct values of $\mathcal{A}_N$,  a shared value $a$ exists in $\mathcal{B}$ such that, $\forall u \in  \mathcal{V}$, $ \exists t \in \mathcal{B}$ satisfying ($t[U] = u) \wedge (t[A_k] = a)$.
Then, $\mathcal{B}$ must form a dense subgraph $\mathcal{G}$ in $\mathbf{G}$.
\end{theorem}
\vspace{-1em}

\begin{proof}
Using the optimization algorithm in Section \ref{sec:computing}, we build $\mathbf{G}$ by scanning all values in $\mathbf{R}$ once. Hence, the block $\mathcal{B}(a, A_k)$ must be found.
Let $\mathcal{G = (V, E)}$ be the subgraph induced by $\mathcal{V}$ in $\mathbf{G}$. Then, $\forall (u_i, u_j) \in \mathcal{E}$, the edge $\sscoreij \geq  \iscorek{i}{j}{k}(a)$. Hence, $\rho_{edge}(\mathcal{G}) = 1.0 $ and $\fscoreg =\frac{\sum_{\forall (\user{i},\user{j}) \in \mathcal{E} } \sscoreij + \sum_{u_i \in \mathcal{V}} \sscorei}{|\mathcal{V}|} \geq \frac{|\mathcal{V}|(|\mathcal{V}|-1) \iscorek{i}{j}{k}(a)}{ |\mathcal{V}|} = (|\mathcal{V}|-1)\iscorek{i}{j}{k}(a)$.
\end{proof}
\vspace{-0.5em}

\noindent \textbf{Observation.} (Effectiveness of ISG+D-Spot) Consider a hidden-densest block $\mathcal{B}(\mathcal{A}_1, ..., \mathcal{A}_{N-1},$ $ \mathcal{A}_N, \mathcal{X})$ of size $|\mathcal{X}| \times ... \times |\mathcal{X}| \times  1$ and $|\mathcal{B}| = |\mathcal{X}|$, i.e., $\mathcal{B}$ is the densest on $A_N$ by sharing the value $a$. Then, assuming the target dimension $U = A_1$ and the fraudulent entities $\mathcal{V}$ are distinct values of $\mathcal{A}_1$,   ISG+D-Spot captures $\mathcal{V}$ more accurately than other algorithms based on tensors (denoted as Tensor+Other Algorithms). \\

\vspace{-1em}
\begin{proof}  Let us consider a non-dense block  $\hat{\mathcal{B}}( \hat{\mathcal{A}}_1, ... , \hat{\mathcal{A}}_N, \mathcal{X})$ of size $|\mathcal{X}| \times ... \times |\mathcal{X}|$, $|\hat{\mathcal{B}}| = |\mathcal{X}|$, and let $\hat{\mathcal{V}}$ denote the distinct values of $\hat{\mathcal{A}}_1$.  Denoting legitimate entities as $ \hat{\mathcal{V}}$ and fraudulent entities as $\mathcal{V}$, we now discuss the difference between ISG+D-Spot and Tensor+Other Algorithms.
\smallskip

\noindent [\textit{Working on the tensor}]. On $\mathbf{R}$, $\hat{\mathcal{B}}$ is not dense and thus $\rho(\hat{\mathcal{B}}, [N]) = 1$. For $\mathcal{B}$, because $\{|\mathcal{B}_1|, ..., |\mathcal{B}_{N-1}|, |\mathcal{B}_N|  \} = \{|\mathcal{X}|, ...,|\mathcal{X}|, 1\}$, we have $\rho(\mathcal{B}, [N]) = \frac{|\mathcal{B}| }{\frac{1}{N} \sum_{n\in [N]} |\mathcal{B}_n|} \approx 1 $ for sufficiently large $N$.

\smallskip

\noindent [\textit{Working on ISG}]. On $\mathbf{G}$, let $\hat{\mathcal{G}}$ denote the subgraph induced by $\hat{\mathcal{V}}$ and $\mathcal{G}$ denote the subgraph formed by $\mathcal{V}$.  We know that $\hat{\mathcal{G}}$, $\mathcal{F}_{\hat{\mathcal{G}}} = 0$, because $\hat{\mathcal{B}}$ does not have any shared values. For $\mathcal{G}$, $\mathcal{F}_{\mathcal{G}} = (\mathcal{V}-1)\iscorek{i}{j}{k}(a)$ according to Theorem \ref{theo1}.
\smallskip


\noindent [\textit{Other Algorithms}]. M-Zoom\cite{MZOOM} and D-Cube \cite{DCUBE} are known to find blocks that are at least $1/N$ of the optimum in term of $\rho$ on $\mathbf{R}$ ($\frac{1}{N}$-Approximation guarantee).    \smallskip

\noindent [\textit{D-Spot}]. In Section \ref{sec:acc}, we will show that the subgraph detected by D-Spot is at least $1/2$ of the optimum in term of $\fscoreg$ on ISG ($\frac{1}{2}$-Approximation guarantee). \smallskip

\noindent In summary, Tensor+Other Algorithms  vs. ISG+D-Spot corresponds to:
\begin{displaymath}
\begin{split}
& \bigg ( \rho(\hat{\mathcal{B}}, [N]) = 1 \ | \ \rho(\mathcal{B}, [N]) \approx 1 \ +\ (\frac{1}{N}\text{-Approximation}) \bigg ) \\
 vs. \ &  \bigg (  \mathcal{F}_{\hat{\mathcal{G}}} = 0 \ | \ \mathcal{F}_{\mathcal{G}} = (\mathcal{V}-1)\iscorek{i}{j}{k}(a) \ + \ (\frac{1}{2}\text{-Approximation})   \bigg ) 
\end{split}
\end{displaymath}
Therefore, ISG+D-Spot catches fraudulent entities within hidden-densest blocks more accurately than Tensor+Other Algorithms.
\end{proof}

From the observation, ISG+D-Spot can effectively detect hidden-densest blocks. 
Similarly, when $\mathcal{B}$ becomes denser, the $\mathcal{G}$ formed by  $\mathcal{B}$ will also be much denser, and thus  ISG+D-Spot will be more accurate in detecting the densest block.

\vspace{-0.5em}
\subsection{Accuracy Guarantee of D-Spot} \label{sec:acc}

For brevity, we use $[\mathcal{V}]$ to denote a subgraph induced by the set of nodes $\mathcal{V}$.\\

\vspace{-1em}
\begin{theorem} \label{gua1}  
(Algorithm 1 Guarantee). 
Given $\mathbf{G} = (\mathbf{V}, \mathbf{E}) $,  let $\mathcal{G}$s $= \{\mathcal{G}_1, ... , \mathcal{G}_n\}$ denote the connected components of $\mathbf{G}$. Let $\mathcal{F}_i^{opt}$ denote the optimal $\mathcal{F}$ on $\mathcal{G}_i$, i.e.,
$\nexists \mathcal{G}' \subseteq \mathcal{G}_i$ satisfying $\mathcal{F}_{\mathcal{G}'} > \mathcal{F}_i^{opt}$. 
Then, if  $\mathcal{F}_n^{opt}$ is the maximal value of \{$\mathcal{F}_1^{opt}$, ..., $\mathcal{F}_n^{opt}$ \}, $\mathcal{F}_n^{opt}$ must be the optimum in terms of $\mathcal{F}$ on $\mathbf{G}$.
\end{theorem}

\vspace{-1em}
\begin{proof}  
Given any two sets of nodes $\mathcal{V}_1$ and $\mathcal{V}_2$ and assuming there are no edges connecting $\mathcal{V}_1$ and $\mathcal{V}_2$, we assume that $\mathcal{F}_{[\mathcal{V}_1]} > \mathcal{F}_{[\mathcal{V}_2]} \Rightarrow \frac{c_1}{|\mathcal{V}_1|} >  \frac{c_2}{|\mathcal{V}_2|} \Rightarrow c_1|\mathcal{V}_2| > c_2|\mathcal{V}_1| $. Then, 
\begin{displaymath}
\begin{split}
\mathcal{F}_{[\mathcal{V}_1]} - \mathcal{F}_{[\mathcal{V}_1 \cup \mathcal{V}_2]} \Rightarrow
\frac{c_1}{|\mathcal{V}_1|}  - \frac{c_1 + c_2 }{|\mathcal{V}_1| + |\mathcal{V}_2|}\\
\Rightarrow \frac{c_1 |\mathcal{V}_2| - c_2 |\mathcal{V}_1|}{ |\mathcal{V}_1| ( |\mathcal{V}_1| +  |\mathcal{V}_2| )}
>  \frac{c_2 |\mathcal{V}_1| - c_2 |\mathcal{V}_1|}{ |\mathcal{V}_1| ( |\mathcal{V}_1| +  |\mathcal{V}_2| )} = 0.
\end{split}
\end{displaymath}
Thus, for any $\mathcal{V}_1$ and $\mathcal{V}_2$ that are not connected by any edges, 
it follows that $\mathcal{F}_{[\mathcal{V}_1 \cup  \mathcal{V}_2]} \leq \max(\mathcal{F}_{[\mathcal{V}_1]},  \mathcal{F}_{[\mathcal{V}_2]} ) $ (Conclusion 1).

In $\mathcal{G}_n =(\mathcal{V}_n, \mathcal{E}_n)$, we use $\hat{\mathcal{V}}$ to denote the set of nodes satisfying $\hat{\mathcal{V}} \subseteq \mathcal{V}_n  $ and $\mathcal{F}_{[\hat{\mathcal{V}}]} = \mathcal{F}_n^{opt}$.  
Let $\mathcal{V}'$ be a set of nodes satisfying $\mathcal{V}' \subset \mathbf{V}$ and $\mathcal{V}' \cap \hat{\mathcal{V}} = \emptyset$.
Now, let us consider two conditions.

First, if $\mathcal{V}' \subset \mathcal{V}_n$, then $\mathcal{F}_{[\mathcal{V}']} \leq \mathcal{F}_{[\hat{\mathcal{V}}]}$ and  $\mathcal{F}_{[\mathcal{V}' \cup \hat{\mathcal{V}}]} \leq \mathcal{F}_{[\hat{\mathcal{V}}]}$ because $\mathcal{F}_{[\hat{\mathcal{V}}]}$ is the optimum on $\mathcal{G}_n$.

Second, if $\mathcal{V}' \cap  \mathcal{V}_n = \emptyset $, 
then $\mathcal{F}_{[\mathcal{V}']} \leq \mathcal{F}_{[\hat{\mathcal{V}}]}$ and  $\mathcal{F}_{[\mathcal{V}' \cup \hat{\mathcal{V}}]} \leq \mathcal{F}_{[\hat{\mathcal{V}}]}$ by Conclusion 1 and because $\mathcal{F}_{[\hat{\mathcal{V}}]}$ is the maximum of \{$\mathcal{F}_1^{opt}$, ..., $\mathcal{F}_n^{opt}$\}.

If $\mathcal{V}' \cap  \mathcal{V}_n \neq \emptyset $, then $\mathcal{V}'$ can be divided into two parts conforming with the two conditions stated above.  

Therefore, $\nexists \mathcal{V}' \subset \mathbf{V}$ satisfies $\mathcal{F}_{[\mathcal{V}']} > \mathcal{F}_{[\hat{\mathcal{V}}]}$ or  $\mathcal{F}_{[\mathcal{V}' \cup \hat{\mathcal{V}}]} > \mathcal{F}_{[\hat{\mathcal{V}}]}$.  
We can conclude that  $\mathcal{F}_n^{opt}  = \mathcal{F}_{[\hat{\mathcal{V}}]}$ must be the optimum in terms of $\mathcal{F}$ on $\mathbf{G}$.  
\end{proof}

\begin{theorem} \label{gua2}
(Algorithm 2 Guarantee). 
Given a graph $\subg=(\subv,\sube)$,  let $Q^*$ be a subset of nodes maximizing $\mathcal{F}_{[Q^*]}$ in $\subg$. Let $[Q]$ be the subgraph returned by Algorithm 2 with $\mathcal{F}_{[Q]}$. Then,  $\mathcal{F}_{[Q]} \geq \frac{1}{2}  \mathcal{F}_{[Q^*]}$.
\end{theorem}

\begin{proof}
Consider the optimal set $Q^*$. We know that, $\forall u\in Q^*$,  $w(u, [Q^*]) \geq \fscore{[Q^*]}$, because if we remove a node $u$ for which $w(u, [Q^*])<\fscore{[Q^*]}$,
\[\mathcal{F}'=\frac{|Q^*|\fscore{[Q^*]}-w(u, Q^*)}{|Q^*|-1}
> \frac{|Q^*|\fscore{Q^*}-\fscore{Q^*}}{|Q^*|-1}=\fscore{[Q^*]},
 \]
which contradicts the definition of $Q^*$.

Denote the first node that Algorithm 2 removes from $Q^*$ as $\user{i}$, $u_i \in R$, and denote the node set before Algorithm 2 starts removing $R$ as $Q'$. Because $Q^* \subseteq Q'$, we have $w(u_i, [Q^*])\leq w(u_i, [Q'])$. According to Algorithm 2 (line 6), $w(u_i, [Q']) \leq \  2\fscore{[Q']}$ (Eq. 10). Additionally, Algorithm 2 returns the best solution when deleting nodes one by one, and so $\fscore{[Q]}\geq\fscore{[Q']}$. We conclude that
\begin{displaymath}
\fscore{[Q]}\geq \fscore{[Q']}\geq \frac{w(u_i, [Q'])}{2}\geq \frac{w(u_i, [Q^*])}{2}\geq \frac{\fscore{Q^*}}{2}.
\end{displaymath}
\end{proof}

In summary, let $\{\mathcal{G}_1,..., \mathcal{G}_n\}$ be the subgraphs returned by D-Spot, and $\{\mathcal{F}_{\mathcal{G}_1},..., \mathcal{F}_{\mathcal{G}_n}\}$ be the corresponding scores. Then, based on Theorems \ref{gua1} and \ref{gua2}, $\mathcal{F}^{max} = max(\mathcal{F}_{\mathcal{G}_1},..., \mathcal{F}_{\mathcal{G}_n})$ is at least $1/2$ of the optimum in terms of $\mathcal{F}$ on $\mathbf{G}$ ($1/2$-Approximation guarantee).

\begin{table*}[t]
	\caption{Multi-dimensional datasets used in our experiments}\label{tab:detail}
	\centering
	\begin{tabular}{c|c|c|c|c|c|c|c|c|c}
		\toprule
		&Synthetic &\multicolumn{2}{|c|}{TCP Dumps}&\multicolumn{6}{|c}{Review Data}\\
		\midrule  
		& & DARPA & AirForce & YelpChi & YelpNYU & YelpZip & AmaOffice& AmaBaby & AmaTools \\ 
		\midrule  
		Entries (Mass) &10K & 4.6M  & 30K & 67K & 359K & 1.14M & 53K & 160K & 134K   \\
		Dimensions& 7& 3 & 3 & 3 & 3 & 3 & 3 & 3 & 3\\ 
		\bottomrule
	\end{tabular}
	\vspace{-1em}
\end{table*}

\section{Evaluation}
\label{sec:eval}

A series of evaluation experiments were conducted under the following conditions:

\para{Implementation.} We implemented ISG+D-Spot in Python, and conducted all experiments on a server with two 2.20 GHz  Intel(R) CPUs and 64 GB memory.

\para{Baselines.}
We selected several state-of-the-art dense-block detection methods  (M-Zoom~\cite{MZOOM}, M-Biz \cite{MBIZ}, and D-Cube~\cite{DCUBE}) as the baselines (using open source code). To obtain optimal performance, we run three different density metrics from \cite{DCUBE} for each baseline: $\rho$ \textbf{(ari)}, Geometric Average Mass \textbf{(geo)}, and Suspiciousness \textbf{(sus)}.

\para{Suspiciousness Score Setting}. For the baselines, we considered a detected block $\mathcal{B}(\mathcal{A}_1,...,\mathcal{A}_N, \mathcal{X})$ and let $\theta = \rho(\mathcal{B}, N)$. For any unique value $a$ within $\mathcal{B}$, we then set the suspiciousness score of $a$ to $\theta$. If $a$ occurred in multiple detected dense blocks, we chose the one with the maximal value of $\theta$.  Regarding ISG+D-Spot, given a detected subgraph $\subghat = (\hat{\mathcal{V}}, \hat{\mathcal{E}})$, for a unique value $a \in \hat{\mathcal{V}}$, we set the suspiciousness score of $a$ to $w(a, \subghat)$ (Eq. \ref{eq:w}). Finally, we evaluated the ranking of the suspiciousness scores of unique values using the standard area under the ROC curve (AUC) metric.

%

\vspace{-0.5em}
\subsection{Datasets}
Table \ref{tab:detail} summarizes the datasets used in the experiments.

\para{Synthetic} is a series of datasets we synthesized using the same method as in~\cite{crossspot}. 
First, we generated random seven-dimensional relations $\mathbf{R}(A_1, ... ,A_7, X)$, in which $|\mathbf{R}| = 10000$ and the size of $\mathbf{R}$ is $1000 \times 500 \times ... \times 500$. In $\mathbf{R}$, we assume that $A_1$ corresponds to users and the other six dimensions are features. To specifically check the detection performance of the hidden-densest block using each method, we injected a dense block $\mathcal{B}(\mathcal{A}_1,..., \mathcal{A}_7, \mathcal{X})$ into $\mathbf{R}$ five separate times, with each injection assigned a different configuration to generate five datasets. 
For $\mathcal{B}$,  $|\mathcal{B}_1| = 50$ and $|\mathcal{B}| = 500$. We introduce the parameter $\lambda$, which denotes the number of dimensions on which $\mathcal{B}$ is the densest.  For example, when $\lambda = 1$, the size of $\mathcal{B}$ is $50 \times$\textbf{12}$\times 25 ... \times 25$; when $\lambda = 5$, the size of $\mathcal{B}$ is $50 \times$\textbf{12}$\times ... \times$ \textbf{12} $\times 25$. Obviously, $\rho(\mathcal{B}, 7) >  \rho(\mathbf{R}, 7)$ and $\mathcal{B}$ is the hidden-densest block when $\lambda$ is small. Finally, we labeled the users within $\mathcal{B}$ as ``fraud''.

\para{Amazon}~\cite{amazon_data}. AmaOffice, AmaBaby, and AmaTools are three collections of reviews about office products, baby-related products, and tool products, respectively, on Amazon. They can be modeled using the relation $\mathbf{R}(user, product,  timestamp,  X)$.  For each entry $t \in \mathbf{R}$, $t= (u, p, t, x)$ indicates a review $x$ that user $u$ reviewed product $p$ at time $t$. According to the specific cases of fraud discovered by previous studies \cite{COPYCATCH,crossspot}, fraudulent groups  usually exhibit suspicious  synchronized behavior in social networks. For instance, a large group of users may surprisingly review the same group of products over a short period. Thus, we use a similar method as in \cite{COPYCATCH, crossspot, MZOOM,DCUBE}. We use a dense block $\mathcal{B}$ to represent the synchronized behavior, where $\mathcal{B}(user, product, timestamp, \mathcal{X})$ has a size of $200 \times 30 \times 1$. In total, we injected four such blocks $\mathcal{B}$ with a mass randomly selected from $[1000, 2000]$. The users in the injected blocks were labeled as ``malicious.''   

\para{Yelp}~\cite{YELP}. The YelpChi, YelpNYU, and YelpZip datasets \cite{YELP2,YELP} contain restaurant reviews submitted to Yelp. They can be represented by the relation $\mathbf{R}(user, restaurant, date, X)$, where each entry $t = (u, r, d, x)$ denotes a review $x$ by user $u$ of restaurant $r$ on date $d$. 
Note that all three datasets include labels indicating whether or not each review is fake. The detection of malicious reviews or users is studied in \cite{YELP} using text information. In these datasets, we focus on detecting fraudulent restaurants that purchase fake reviews using the three-dimensional features. Intuitively, the more fake reviews a restaurant has, the more suspicious it is. As some legitimate users have the potential of reviewing fraudulent restaurants, we label a restaurant as ``fraudulent'' if it has received more than 40 fake reviews.

\para{DARPA}~\cite{DARPA} was collected by the Cyber Systems and Technology Group in 1998 regarding network attacks in TCP dumps. The data has the form $\mathbf{R}(source IP, target IP, timestamp, X)$. Each entry $t = (IP_1, IP_2 , t)$ represents a connection made from $IP_1$ to $IP_2$ at time $t$. The dataset includes labels indicating whether or not each connection is malicious. In practice, the punishment for malicious connections is to block the corresponding IP address. Thus, we compared the detection performance of suspicious IP addresses. We labeled an IP address as suspicious if it was involved in a malicious connection.   

\para{AirForce}~\cite{AIRFORCE} was used for the KDD Cup 1999, and has also been considered in \cite{MZOOM,DCUBE}.  This dataset includes a wide variety of simulated intrusions in a military network environment. However, it does not contain any specific IP addresses.  According to the cardinality of each dimension, we chose the top-2 features and built the relation $\mathbf{R}(src \ bytes, dst \ bytes, connections, X)$, where src bytes denotes the number of data bytes sent from source to destination and dst bytes denotes the number of data bytes sent from destination to source. The target dimension $U$ was set to be $connections$. Note that this dataset includes labels indicating whether or not each connection is malicious.

\subsection{Speed and Accuracy of D-Spot} \label{sec:speed}
First, we measured the speed and accuracy with which D-Spot detected dense subgraphs in real-world graphs. We compared the performance of D-Spot with that of another dense-subgraph detection algorithm, Fraudar \cite{FRAUDAR}, the extension of \cite{greedy}, which maximizes the density metric by greedily deleting nodes \textbf{one by one}. We used the three Amazon datasets and applied D-Spot and Fraudar to the same bipartite graph built on the first two dimensions, \emph{users} and \emph{products}, where each edge in the graph represents an entry. We measured the wall-clock time (average over three runs) required to detect the top-4 subgraphs. Figure~\ref{fig:2mode} illustrates the runtime and performance of the two algorithms.

\begin{figure}[h]
\setlength{\belowcaptionskip}{-0.2cm} 
	\includegraphics[width=0.8\columnwidth, height = 0.25\columnwidth]{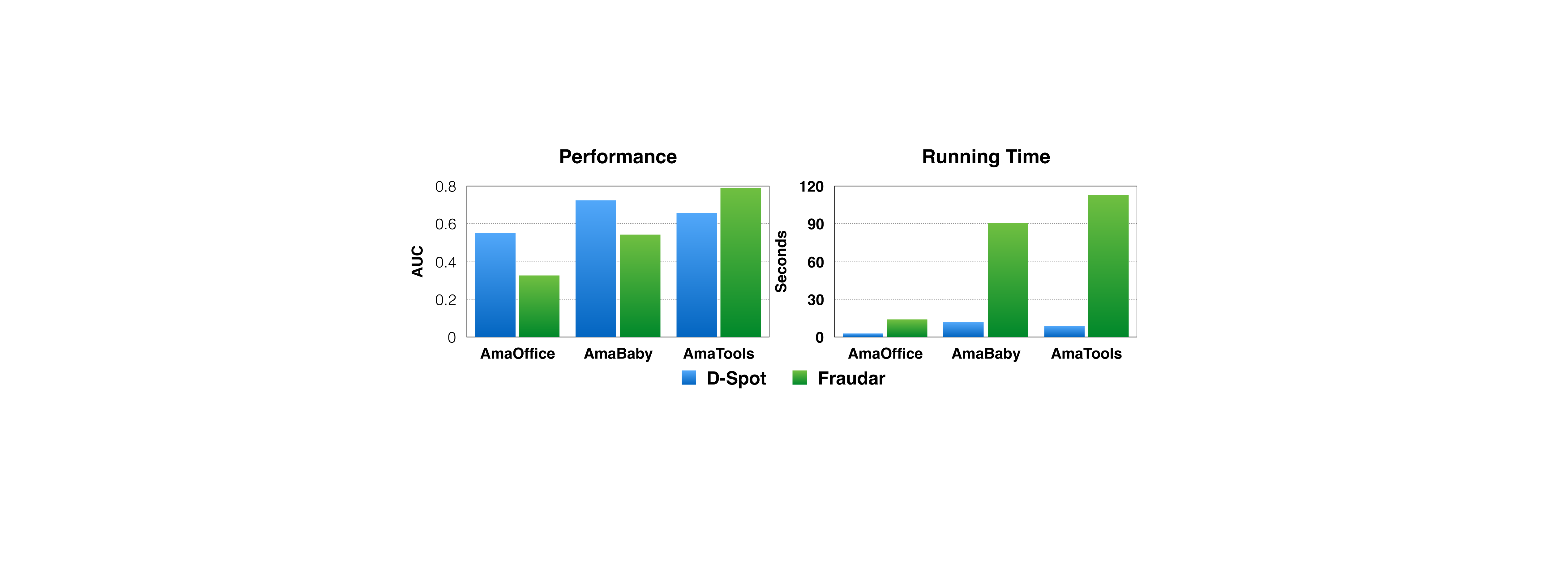}
	\centering
	\vspace{-0.2cm}
	\caption{Comparison of D-Spot and Fraudar  using the Amazon datasets. }
	\label{fig:2mode}
	\vspace{-0.2cm}
\end{figure}

D-Spot provides the best trade-off between speed and accuracy. Specifically, D-Spot is up to  11$\times$ faster than Fraudar. This supports our claim in Section \ref{sec:bads} that the worst-case time complexity of D-Spot ($O( |\mathcal{V}|^2 + |\mathcal{E}|)$) is too pessimistic. 

\subsection{Effectiveness of ISG+D-Spot}
This section illustrates the effectiveness of ISG+D-Spot for detecting fraudulent entities on multi-dimensional tensors. ISG+D-Spot exhibits extraordinary performance compared with the baseline methods (Fraudar is not in the baselines as it only works on the bipartite graph).   

\para{Synthetic.} Table \ref{tab:synthetic} presents the detection performance of each method for the hidden-densest block. We assume that the injected block $\mathcal{B}$ is the hidden-densest block when $\lambda \leq 3$. In detail, ISG+D-Spot achieves extraordinary performance even when $\lambda = 1$, because each instance of value sharing in $\mathcal{B}$ is accurately captured by ISG and  D-Spot, providing a higher accuracy guarantee than the baselines (Theorems 2 and 3). When $\lambda > 3$, the performance of each method improves because the density of $\mathcal{B}$ increases as $\lambda$ increases.

\begin{table}[h]
	\caption{Performance (AUC) on the Synthetic datasets }\label{tab:synthetic}
	\vspace{-1em}
	\centering
	\begin{tabular}{c|ccccc}
		\toprule
		& $\lambda =1  $ & $\lambda =2$& $\lambda =3$&$\lambda =4$&$\lambda =5$ \\
		\midrule  
		\emph{M-Zoom (ari)}&0.5005  &0.5005&0.5000&0.5489 & 0.6567 \\
		\emph{M-Zoom (geo)} & 0.5005  &0.5005 &0.5000&0.6789 &0.7543 \\
		\emph{M-Zoom (sus)}  & 0.7404  & 0.7715 &0.8238&0.9685 &0.9767 \\
		\emph{M-Biz (ari)}   & 0.5005 & 0.5005 &0.5005&0.5005 &0.6834 \\
		\emph{M-Biz (geo)}   & 0.5005 & 0.5005&0.5005&0.6235 &0.7230 \\
		\emph{M-Biz (sus)}   & 0.6916 & 0.7638 & 0.8067&0.9844 &\textbf{0.9948} \\
		\emph{D-Cube (ari)} &0.5005 &0.5005& 0.5005&0.5670  &0.6432 \\
		\emph{D-Cube (geo)} &0.5005 &0.5005 & 0.5340 &0.6876 &0.6542 \\
		\emph{D-Cube (sus)} &\textbf{0.8279} & \textbf{0.8712}&\textbf{0.9148} &\textbf{0.9909} &0.9725  \\
		\midrule  
		ISG+D-Spot & \textbf{0.9843} & \textbf{0.9957} & \textbf{0.9949} &\textbf{1.0000} & \textbf{1.0000}\\
		\bottomrule
		
	\end{tabular}
	\vspace{-0.5em}
\end{table}

\para{Amazon.} Table \ref{tab:amazon} presents the results for catching suspicious users by detecting the top-4 dense blocks on the Amazon datasets. ISG+D-Spot detects synchronized behavior accurately. The typical attack scenario involves a mass of fraudulent users creating massive numbers of fake reviews for a comparatively small group of products over a short period. This behavior is represented by the injected blocks. ISG+D-Spot exhibits robust and near-perfect performance. However, the other baselines produce worse performance on the AmaOffice and AmaBaby datasets, even with the multiple supported metrics.  

\begin{table}[h]
	\caption{Performance (AUC) on the Amazon datasets }\label{tab:amazon}
		\vspace{-1em}
	\centering
	\begin{tabular}{c|cccc}
		\toprule
		& AmaOffice & AmaBaby& AmaTools \\
		\midrule  
		\emph{M-Zoom (ari)}& 0.6795  &0.5894& 0.8689  \\
		\emph{M-Zoom (geo)} & 0.8049  & 0.8049 &1.0000 \\
		\emph{M-Zoom (sus)}  & 0.7553  & 0.6944 &0.6503 \\
		\emph{M-Biz (ari)}   & 0.6328 & 0.5461 &0.8384\\
		\emph{M-Biz (geo)}   & 0.8049 & \textbf{0.8339}& 1.0000 \\
		\emph{M-Biz (sus)}   & 0.7478 & 0.6944 & 0.6503 \\
		\emph{D-Cube (ari)} & 0.7127 &0.5956& 0.6250  \\
		\emph{D-Cube (geo)} &\textbf{0.8115 }&0.7561 & 1.0000  \\
		\emph{D-Cube (sus)} & 0.7412 &0.6190& 0.5907  \\
		\midrule  
		ISG+D-Spot & \textbf{0.8358} & \textbf{0.9995} & \textbf{1.0000}\\
		\bottomrule
	\end{tabular}
	\vspace{-1em}
\end{table}

\para{Yelp.} Table \ref{tab:yelp} reports the (highest) accuracy with which collusive restaurants were detected by each method. In summary, using ISG+D-Spot results in the highest accuracy across all three datasets, because D-Spot applies a higher theoretical bound to the ISG.  
\vspace{-0.5em}
\begin{table}[h]
	\caption{Performance (AUC) on the Yelp datasets}\label{tab:yelp}
		\vspace{-1em}
	\centering
	\begin{tabular}{l|ccc}
		\toprule
		&YelpChi & YelpNYU & YelpZip\\
		\midrule  		  
		\emph{M-Zoom (ari)} & 0.9174 & 0.6669 &0.8859 \\
		\emph{M-Zoom (geo)} & 0.9752 & 0.8826 & 0.9274 \\
		\emph{M-Zoom (sus)} & 0.9831 & \textbf{0.9451} & 0.9426 \\
		\emph{M-Biz (ari)}  & 0.9174 & 0.6669 & 0.8863 \\
		\emph{M-Biz (geo)}  & 0.9757 & 0.8826 &0.9271 \\
		\emph{M-Biz (sus)}  &\textbf{0.9831} & 0.9345 & \textbf{0.9403} \\
		\emph{D-Cube (ari)} &0.5000 & 0.5000& 0.9033  \\
		\emph{D-Cube (geo)} &0.9793 & 0.9223& 0.9376 \\
		\emph{D-Cube (sus)} &0.9810 & 0.9007 &0.9365 \\
		\midrule  
		ISG+D-Spot & \textbf{0.9875}& \textbf{0.9546}& \textbf{0.9529} \\
		\bottomrule
	\end{tabular}
	\vspace{-1em}
	
\end{table}

\para{DARPA.} Table \ref{tab:darpa} lists the accuracy of each method for detecting the source IP and the target IP. ISG+D-Spot assigns each IP address a specific suspiciousness score. We chose a detected IP with the highest score and found that the IP participated in more than 1M malicious connections. The top ten suspicious IPs were all involved in more than 10k malicious connections. Thus, using ISG+D-Spot would enable us to crack down on these malicious IP addresses in the real world.

\begin{table}[h]
	\caption{Performance (AUC) on the DARPA dataset}\label{tab:darpa}
		\vspace{-1em}
	\centering
	\begin{tabular}{c|cc}
		\toprule
		Dataset &\multicolumn{2}{|c}{DARPA}\\
		\midrule  
		&U = source IP & U = target IP \\
		\midrule  
		\emph{M-Zoom (ari)}& 0.5649  &0.5584 \\
		\emph{M-Zoom (geo)} &0.7086  &\textbf{0.5714}  \\
		\emph{M-Zoom (sus)}  &0.6989  &0.3878 \\
		\emph{M-Biz (ari)}   &0.5649 &0.5584  \\
		\emph{M-Biz (geo)}   &\textbf{0.7502} &0.5679  \\
		\emph{M-Biz (sus)}   &0.6989 &0.3878  \\
		\emph{D-Cube (ari)} &0.3728 &0.5323  \\
		\emph{D-Cube (geo)} &0.4083 &0.3926   \\
		\emph{D-Cube (sus)} &0.4002 &0.3720  \\
		\midrule  
		ISG+D-Spot & \textbf{0.7561} & \textbf{0.8181}\\
		\bottomrule	
	\end{tabular}
	\vspace{-0.5em}
\end{table}

\para{AirForce.} As this dataset does not contain IP addresses, we set the target dimension $U = connections$. We randomly sample 30k connections from the dataset \cite{AIRFORCE} three times. Table \ref{tab:airforce} lists the accuracy of each method on samples 1--3. Malicious connections form dense blocks on the two-dimensional features. The results demonstrate that ISG+D-Spot effectively detected the densest blocks.

\vspace{-0.5em}
\begin{table}[h]
	\caption{Performance (AUC) on the AirForce dataset}\label{tab:airforce}
		\vspace{-1em}
	\centering
	\begin{tabular}{c|cc|cc}
		\toprule
		&Sample 1 & Sample 2   & Sample 3\\
		\midrule  
		\emph{M-Zoom (ari)}&0.8696  &0.8675  & 0.8644 \\
		\emph{M-Zoom (geo)} &0.9693  & 0.9693  &0.9738 \\
		\emph{M-Zoom (sus)}  &0.9684  & 0.9683  & 0.9726 \\
		\emph{M-Biz (ari)}   &0.9038 & 0.8848  & 0.8885\\
		\emph{M-Biz (geo)}   &0.9694 & 0.9693  & 0.9741 \\
		\emph{M-Biz (sus)}   &0.9684 & 0.9683  & 0.9726 \\
		\emph{D-Cube (ari)} &\textbf{0.9824} &\textbf{0.9823}& 0.9851  \\
		\emph{D-Cube (geo)} &0.9695 &0.9691  &\textbf{0.9862} \\
		\emph{D-Cube (sus)} &0.9697 &0.9692  & 0.9737\\
		\midrule  
		ISG+D-Spot & \textbf{0.9835} & \textbf{0.9824}& \textbf{0.9877}\\
		\bottomrule
		
	\end{tabular}
	\vspace{-1em}
\end{table}

\vspace{-0.5em}
\subsection{Scalability} 
As mentioned in Sec.\ref{sec:com}, the ISG $\mathbf{G}$s built using real-world tensors are typically sparse, as value sharing should only conspicuously appear in fraudulent entities.
 We implemented $\mathbf{G}$ on the three Amazon datasets (details in Figure \ref{fig:time}). 
The edge densities of $\mathbf{G}$ are quite low (lower than 0.06) across all datasets, which indicates that the worst-case time complexity discussed in Section \ref{sec:com} rarely occurs. Figure \ref{fig:time} reports the runtime of ISG+D-Spot on the three Amazon datasets, where the number of edges was varied by subsampling entries in the dataset. In practice, $|\mathbf{E}|$ increases near-linearly with the mass of the dataset. Additionally, because the time complexity of  ISG+D-Spot is linear with respect to $|\mathbf{E}|$, 
ISG+D-Spot exhibits near-linear scaling with the mass of the dataset. Figure \ref{fig:time} demonstrates our conclusion.

\begin{figure}[h] 
    \includegraphics[width=0.8\columnwidth, height=0.3\columnwidth]{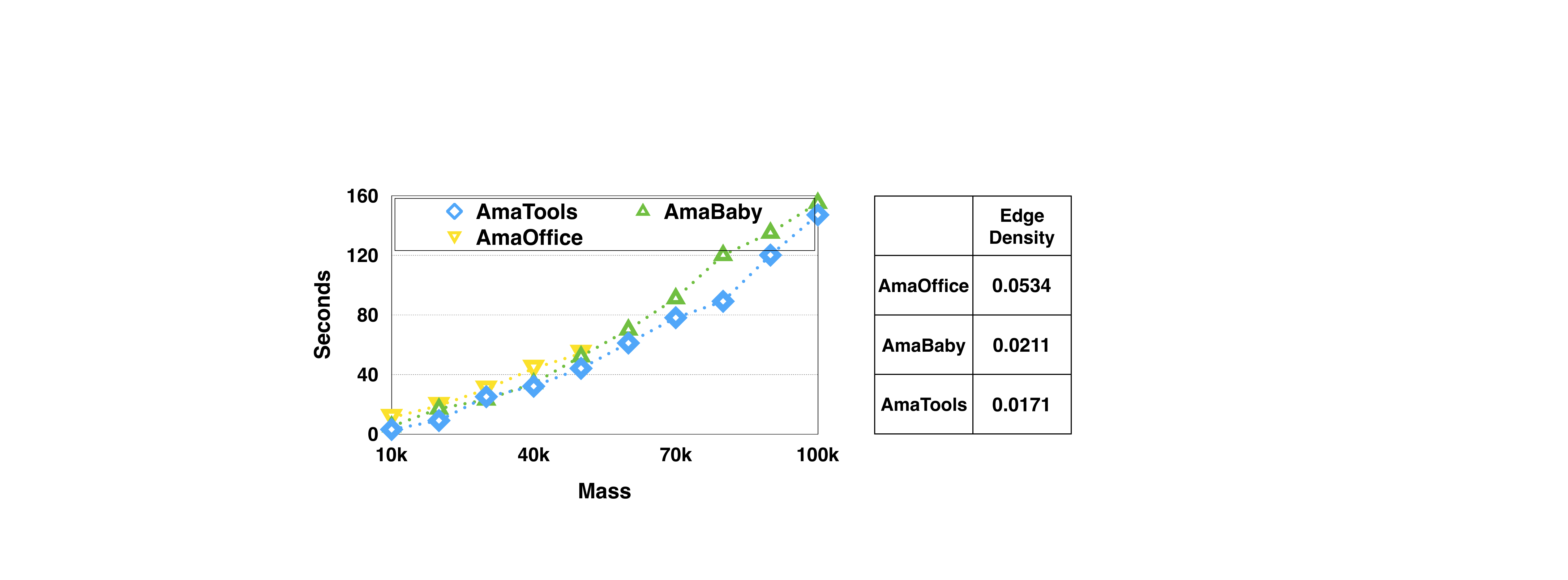}
    \vspace{-1em}
    \caption{ ISG+D-Spot runs in near-linear time with respect to the mass of the dataset.}
       \label{fig:time}
    \vspace{-1.5em}
\end{figure}

\subsection{Feature Prioritization}

This section is to demonstrate that ISG+D-Spot is more robust to resist noisy features than existing approaches. ISG automatically weighs each feature and continuously accumulates value sharing by one scan of the tensor, and D-Spot amounts to finds entities with the maximum of value sharing. We conducted the following experiment to demonstrate our conclusion.

\para{Registration} is a dataset derived from an e-commerce company, in which each record contains two crucial features, IP subnet and phone prefix, and three noisy features, IP city, phone city, and timestamp.  The dataset also includes labels showing whether or not the account is a ``zombie'' account. Thus, it can be formulated as $\mathbf{R}($\textit{accounts, IP, phone, IP \ city, phone \ city, timestamp}, $X)$. To compare the detection performance of malicious accounts,  we applied each method on various $\mathbf{R}$ by successively appending 1--5 features to $\mathbf{R}(accounts, X)$.

\vspace{-1em}
\begin{table}[h]
	\caption{Performance (AUC) on the Registration dataset. `C' represents `crucial feature' and `N' represents `noisy feature' }\label{tab:regi}
	\vspace{-0.4cm}
	\centering
	\begin{tabular}{c|ccccc}
		\toprule
		& 1C   & 2C   & 2C+1N & 2C+2N & 2C+3N\\
		\midrule  
		\emph{M-Zoom (ari)}&0.5000  &0.5000  & 0.5031 & 0.5000 & 0.5430   \\
		\emph{M-Zoom (geo)} &0.7676  & 0.8880  & \textbf{0.8827} & \textbf{0.8744} & \textbf{0.8439}  \\
		\emph{M-Zoom (sus)}  &0.7466  & 0.8328  &0.4009 & 0.4878& 0.4874\\
		\emph{M-Biz (ari)}   &0.5000 & 0.5000  & 0.5000 & 0.5000&0.5004  \\
		\emph{M-Biz (geo)}   &\textbf{0.7677} & 0.8842  & 0.8827 & 0.8744& 0.8439  \\
		\emph{M-Biz (sus)}   &0.7466 & 0.8328  & 0.4009 & 0.4878 & 0.4874  \\
		\emph{D-Cube (ari)} &0.7073 &0.8189 & 0.8213 & 0.8295 & 0.7987  \\
		\emph{D-Cube (geo)} &0.7073 & \textbf{0.9201}  &0.8586 &  0.8312 & 0.7324 \\
		\emph{D-Cube (sus)} &0.7522 &0.8956  &0.7877 & 0.7642 & 0.7080 \\
		\midrule  
		ISG+D-Spot & \textbf{0.7699} & \textbf{0.9946}& \textbf{0.9935} &\textbf{0.9917} &\textbf{0.9859} \\
		\bottomrule
		
	\end{tabular}
	\vspace{-1em}
	
\end{table}

Table \ref{tab:regi} gives the variation of each method with regard to the added noisy features (3-- 5 dimensions). As each account only possesses one entry, $\mathbf{R}$ is quite sparse. We found that existing methods usually miss small-scale instances of value sharing because their density is close to the legitimate range on $\mathbf{R}$. For example, a 51-member group sharing a single IP subnet was missed by the baseline methods. However, ISG amplifies each instance of value sharing through its information-theoretic and graph features,  allowing D-Spot to accurately capture fraudulent entities.

\section{Conclusions and future work}
\label{sec:conclusion}

In this paper, we novelly identified dense-block detection with dense-subgraph mining, by modeling a tensor in ISG.
Additionally, we propose a multiple dense-subgraphs detection algorithm that is faster and can be computed in parallel.
In future, ISG + D-Spot will be implemented on Apache Spark \cite{SPARK} to support very large tensors. 

\bibliographystyle{ACM-Reference-Format}
\balance 
\bibliography{reference}

\end{document}